\documentclass[11pt]{article}
\usepackage[margin=1in]{geometry}
\usepackage{hyperref}
\usepackage{amsmath,amsfonts,amssymb,amsthm,commath}
\usepackage{enumitem}
\usepackage{framed}
\usepackage{xspace}
\usepackage{microtype}
\usepackage[round]{natbib}
\usepackage{cleveref}
\usepackage[dvipsnames]{xcolor}
\usepackage{graphicx}
\usepackage{authblk}
\usepackage{subcaption}

\def\ddefloop#1{\ifx\ddefloop#1\else\ddef{#1}\expandafter\ddefloop\fi}
\def\ddef#1{\expandafter\def\csname bb#1\endcsname{\ensuremath{\mathbb{#1}}}}
\ddefloop ABCDEFGHIJKLMNOPQRSTUVWXYZ\ddefloop
\def\ddef#1{\expandafter\def\csname c#1\endcsname{\ensuremath{\mathcal{#1}}}}
\ddefloop ABCDEFGHIJKLMNOPQRSTUVWXYZ\ddefloop

\DeclareMathOperator*{\argmin}{arg\,min}
\DeclareMathOperator*{\argmax}{arg\,max}

\newcommand{\eps}{\ensuremath{\epsilon}}
\newcommand{\veps}{\ensuremath{\varepsilon}}

\newcommand\poly{\ensuremath{\operatorname{poly}}}

\newcommand\copt{\ensuremath{a}}
\newcommand\Copt{\ensuremath{\cA}}
\newcommand\Aopt{\ensuremath{A}}
\newcommand\cAopt{\ensuremath{\tilde{A}}}

\newcommand\greedy{\texttt{greedy}\xspace}
\newcommand\core{\ensuremath{\mathrm{core}}\xspace}
\newcommand\kmpp{\texttt{kmeans++}\xspace}
\def\select{\texttt{select}\xspace}
\def\selectpp{\texttt{select++}\xspace}
\def\selectball{\texttt{select-ball}\xspace}

\def\selectsgd{\texttt{select-sgd}\xspace}
\def\guessball{\texttt{guess-ball}\xspace}
\def\bfo{\mathbf{1}}
\def\Pr{\textup{Pr}}
\newcommand{\ip}[2]{\left\langle #1, #2 \right \rangle}

\newtheorem{theorem}{Theorem}[section]
\newtheorem{lemma}{Lemma}[section]

\newtheorem{condition}{Condition}
\crefname{condition}{condition}{conditions}
\crefname{claim}{claim}{claims}

\newcommand\bad{\ensuremath{\mathsf{Bad}}}
\newcommand\good{\ensuremath{\mathsf{Good}}}
\newcommand\constgood{\ensuremath{\kappa_{\operatorname{3}}}\xspace}
\newcommand\conststop{\ensuremath{\kappa_{\operatorname{4}}}\xspace}
\newcommand\constcore{\ensuremath{\kappa_{\operatorname{2}}}\xspace}
\newcommand\constcorelb{\ensuremath{\kappa_{\operatorname{1}}}\xspace}

\def\NP{NP}
\def\dist{D}

\makeatletter
\renewcommand{\paragraph}{%
  \@startsection{paragraph}{4}%
  {\z@}{1ex \@plus 1ex \@minus .2ex}{-1em}%
  {\normalfont\normalsize\bfseries}%
}
\makeatother

\title{Greedy bi-criteria approximations for $k$-medians and $k$-means}
\author[1]{Daniel Hsu}
\author[2]{Matus Telgarsky}
\affil[1]{Columbia University, New York, NY}
\affil[2]{University of Michigan, Ann Arbor, MI}

\begin{document}
\maketitle

\begin{abstract}
  This paper investigates the following natural greedy procedure for
  clustering in the bi-criterion setting:
  iteratively grow a set of centers, in each round adding the center from a
  candidate set that maximally decreases clustering cost.
  In the case of $k$-medians and $k$-means, the key results are as follows.
  \begin{itemize}
    \item
      When the method considers all data points as candidate centers,
      then selecting $\mathcal{O}(k\log(1/\varepsilon))$ centers
      achieves cost at most $2+\varepsilon$ times the optimal cost with $k$ centers.
    \item
      Alternatively, the same guarantees hold
      if each round samples $\mathcal{O}(k/\varepsilon^5)$ candidate centers
      proportionally to their cluster cost
      (as with \texttt{kmeans++}, but holding centers fixed).
    \item
      In the case of $k$-means, considering an augmented set of
      $n^{\lceil1/\varepsilon\rceil}$
      candidate centers gives $1+\varepsilon$ approximation with
      $\mathcal{O}(k\log(1/\varepsilon))$ centers,
      the entire algorithm taking
      $\mathcal{O}(dk\log(1/\varepsilon)n^{1+\lceil1/\varepsilon\rceil})$ time, where $n$ is the number of
      data points in $\mathbb{R}^d$.
    \item
      In the case of Euclidean $k$-medians, generating a candidate set
      via $n^{\mathcal{O}(1/\varepsilon^2)}$ executions of stochastic
      gradient descent with adaptively determined
      constraint sets will once again give approximation $1+\varepsilon$ with
      $\mathcal{O}(k\log(1/\varepsilon))$ centers in
      $dk\log(1/\varepsilon)n^{\mathcal{O}(1/\varepsilon^2)}$ time.
  \end{itemize}
  Ancillary results include:
  guarantees for cluster costs based on powers of metrics;
  a brief, favorable empirical evaluation against \texttt{kmeans++};
  data-dependent bounds allowing $1+\varepsilon$ in the first two bullets above,
  for example with $k$-medians over finite metric spaces.
\end{abstract}

\thispagestyle{empty}
\newpage
\setcounter{page}{1}

\section{Introduction}

Consider the task of covering or clustering a set of $n$ points $X$
using centers from a set $\cY$.
A solution $C\subseteq \cY$ must balance two competing criteria:
its size $|C|$ should be small, as should its cost
\[
  \phi_X(C) \ := \ \sum_{x\in X} \min_{y\in C} \Delta(x,y) \,,
\]
where $\Delta : \cX \times \cY \to \bbR_+$ is a non-negative function defined on $\cY$ and a superset $\cX\supseteq X$.

Amongst many conventions for balancing these two criteria,
perhaps the most prevalent is to fix a reference solution $\Copt$ with $k := |\Copt|$ centers,
and to seek a solution $C$ which minimizes approximation ratio $\phi_X(C) / \phi_X(\Copt)$ while constraining $|C| = k$.
Problems of this type are generally \NP-hard: for example, the $k$-means
problem, where $\Delta(x,y) := \|x-y\|_2^2$ in Euclidean space, and the metric
$k$-medians problem, where $\Delta(x,y) := D(x,y)$ for some metric $D$ over finite
$\cX = \cY$, are each \NP-hard to approximate within some constant factor larger
than one~\citep{jain2002newgreedy,awasthi2015hardness}.

On the other hand,
if $|C|$ is allowed to slightly exceed $|\Copt|$,
the task of approximation becomes much easier.
Returning to the example of $k$-means,
for any $\veps > 0$,
increasing the center budget to $|C| \leq k\ln(1/\veps)$
grants the existence of algorithms with approximation factor $1+\veps$
while taking time polynomial in the size of the input $X$
\citep{makarychev2015bicriteria}.

The classical problem of set cover is similar:
it is \NP-hard, but its natural greedy algorithm finds a cover of size $|C| =
\lceil k \ln(n) \rceil$ whenever one of size $k$ exists~\citep{johnson1974approximation}.
The analogous greedy iterative procedure for $\phi_X$ --- which incrementally adds elements from $\cY$ to maximally decrease $\phi_X$ ---
is the basis of this paper and all its algorithms,
but with one twist: the set of centers in each round, $Y_i$,
is adaptively chosen by a routine called \select.
Instantiating \select in various ways yields the following results.

  \paragraph{Results for $k$-means.}
    This problem takes $\cX = \cY = \bbR^d$
    and $\Delta(x,y) = \|x-y\|_2^2$.
    \begin{itemize}[leftmargin=*]
      \item
        When \select returns all of $X$,
        then $\cO(k\log(1/\veps))$ centers suffice
        to achieve approximation factor $(1+\veps)(1+\constcorelb)$, where $\constcorelb \in [0,1]$
        is a problem-dependent constant (cf.~\Cref{fact:gkm}).
        By contrast, the main competing method \kmpp currently achieves approximation factor $2+\veps$
        with $\cO(k/\veps^2)$ centers \citep{aggarwal2009adaptive,wei2016constant}.
        (A lower bound on the number of centers in this regime is not known;
        when exactly $k$ centers are used, the approximation factor is below $\ln(k)$ only with exponentially
        small probability \citep{brunsch2013badinstance}.)
      \item
        When \select returns $\cO(k/\veps^5)$ points from $X$ subsampled similarly to \kmpp \citep{arthur2007kmeans++},
        once again $\cO(k \log(1/\veps))$ centers suffice
        but with a slightly worse approximation factor $(1+\veps)(1+\constcore)$, where $\constcore\in[\constcorelb, 1+\veps]$
        is another problem-dependent constant (cf.~\Cref{fact:gkm++}).
      \item
        When \select returns the means of all subsets of $X$ of size $\lceil 1/\veps \rceil$,
        then $\cO(k \log(1/\veps))$ centers suffice for approximation factor $(1+\veps)^2$ (cf.~\Cref{fact:kmeans-1+eps}).
        Thus the method requires time $\cO(kd\log(1/\veps)n^{1+\lceil 1/\veps\rceil})$,
        improving upon a running time $n^{\cO(\log(1/\veps)/\veps^2)}$ with
        $\cO(k\log(1/\veps))$ centers,
        due to \citet{makarychev2015bicriteria},
        whose algorithm randomly projects to $\cO(\log(n)/\veps^2)$ dimensions,
        then constructs extra candidate centers via a gridding argument
        \citep{matousek2000approximate},
        then rounds an LP solution,
        and finally lifts the resulting partition back to $\bbR^d$.
        The local search method analyzed by~\citet{bandyapadhyay2016variants}
        returns a solution with approximation factor $1+\veps$ using
        $(1+\veps)k$ centers, but its running time is exponential in
        $(1/\veps)^d$.
        Lastly, assuming the instance satisfies a certain separation condition with parameter
        $\kappa > 0$, a running time of $\cO(n^3)(k\log(n))^{\poly(1/\veps,1/\kappa)}$
        is also possible~\citep{awasthi2010stability}.
    \end{itemize}

  \paragraph{Results for generalized $k$-medians.}
    Variants of the preceding results hold in the following generalized setting,
    versions of which appear elsewhere in the literature
    \citep[e.g.,][]{arthur2007kmeans++}:
    $\Delta(x,y) = \dist(x,y)^p$ for metric $\dist$ on space $\cX = \cY$ (not necessarily infinite), and $p \geq 1$.
    \begin{itemize}[leftmargin=*]
      \item
        The earlier \Cref{fact:gkm,fact:gkm++} go through in this setting still with
        $\cO(k\log(1/\veps))$ centers,
        but respectively granting approximation ratios $(1+\constcorelb)^p$ and $(1+\constcore)^p$ (cf.~\Cref{fact:gkm,fact:gkm++}).
        A notable improvement in this regime is the case of $p=1$ with finite metrics, where $\constcorelb = 0$.
        An approximation factor of $1+\veps$ was obtained in prior work for finite metrics with exactly $k$ centers,
        however requiring separation with a parameter $\kappa > 0$,
        and with an algorithm whose running time is
        $(nk)^{\poly(1/\veps,1/\kappa)}$~\citep{awasthi2010stability}.
      \item
        Achieving approximation ratio $1+\veps$ with $\cO(k\log(1/\veps))$ centers is again
        possible when $D$ is induced by a norm in $\cX = \cY = \bbR^d$.
        If the norm is Euclidean and $p=1$, it suffices to generate candidate
        centers with $n^{3 + \lceil 1/\veps^2\rceil}$ executions of projected
        stochastic gradient descent with adaptively determined constraint
        sets (cf.~\Cref{fact:select-sgd});
        for other norms or exponents $p$, $\cO(n^3 \veps^{-d})$ candidate centers
        need to be sampled (cf.~\Cref{fact:gkm-ball}).
        Existing approximation schemes that only use $k$ centers (for either Euclidean $k$-medians and
        $k$-means) have complexity
        that is either exponential in
        $k$~\citep{kumar2004simple,kumar2005lineartime,feldman2007ptas} or more than exponential in
        $d$~\citep{kolliopoulos1999nearlylinear,cohenaddad2016localsearch,friggstad2016localsearch}.
    \end{itemize}

\paragraph{Related works.}
Analysis of greedy methods is prominently studied in the context of maximizing
submodular functions~\citep{nemhauser1978analysis}, and the recent literature
offers many techniques for efficient
implementation~\citep[e.g.,][]{ashwinkumar2014fast,buchbinder2015comparing}.
It is most natural to view the objective function to be minimized in the present
work as a supermodular function, as opposed to viewing its negation as a
submodular function.
These different viewpoints lead to different approximation results, even for the
same greedy scheme.
Moreover, the specific objective function considered in this work has additional
structure that permits computational speedups (cf.~\Cref{sec:gkm++}) not
generally available for other supermodular objectives.
A more detailed discussion is presented in \Cref{sec:supermod}.

In the context of Euclidean $k$-medians and $k$-means problems
(where $\cX = \cY = \bbR^d$ and $\Delta(x,y) = \norm{x-y}_2^p$ for $p\in\cbr[0]{1,2}$), the
present work is related to the study of bi-criteria approximation algorithms,
which find a solution using $\beta k$ centers whose cost is at most $\alpha$
times the cost of the best solution using $k$ centers.
For any $\veps \in (0,1)$, the factors $\alpha = 2+\veps$ and $\beta =
\poly(1/\veps)$ are achievable for both the $k$-medians
problem~\citep{lin1992approximation} and the $k$-means
problems~\citep{aggarwal2009adaptive,wei2016constant} by
$\poly(n,d,k,1/\veps)$-time algorithms that only select centers from among the
data points $X$.
The bi-criteria approximation methods of
\citet{makarychev2015bicriteria} and \citet{bandyapadhyay2016variants} for
$k$-means are
already discussed above, as are proper approximation schemes for $k$-medians
and $k$-means~\citep{kumar2004simple,kumar2005lineartime,feldman2007ptas,kolliopoulos1999nearlylinear,cohenaddad2016localsearch,friggstad2016localsearch}.

\paragraph{Organization.}
The generic greedy scheme is presented in \Cref{sec:greedy}.
Generalized $k$-medians problems are discussed in \Cref{sec:gkm}.
Lastly, experiments with $k$-means constitute \Cref{sec:exp}.

\paragraph{Notation.}
The set of positive integers $\cbr[0]{1,2,\dotsc,N}$ is denoted by $[N]$, and
$[x]_+ := \max\{0,x\}$ for $x \in \bbR$.
The reference solution $\Copt := \cbr[0]{ \copt_1, \copt_2, \dotsc, \copt_k }$
of cardinality $k$ is treated as fixed in each discussion, however only in some
circumstances is it optimal.
This solution $\Copt$ partitions $X$ into $\Aopt_1, \Aopt_2, \dotsc, \Aopt_k$,
where $\Aopt_j := \cbr[0]{ x \in X : \argmin_{j \in [k]} \Delta(x,\copt_j) = j
}$, breaking $\argmin$ ties using any fixed, deterministic rule.
Observe that $\phi_X(C) = \sum_{j=1}^k \phi_{\Aopt_j}(C)$ for any $C \subseteq
\cY$, and $\phi_X(\Copt) = \sum_{j=1}^k \sum_{x \in \Aopt_j} \Delta(x,\copt_j)$.
The mean of a finite subset $A \subseteq \bbR^d$ is denoted by $\mu(A) :=
\sum_{x \in A} x / |A|$.

\section{Greedy method}
\label{sec:greedy}

The greedy scheme is presented in \Cref{fig:alg}.
As discussed before, it greedily adds a new center in each round
so as to maximally decrease cost.  A routine \select provides the candidate
centers in each round, and the minimization over these candidates need only be solved to accuracy
$1+\tau$.

\begin{figure}[t]
  \begin{framed}
    {\centering\textbf{Algorithm} \greedy\par}
    \begin{itemize}[leftmargin=*,itemsep=0ex]
      \item[] \textbf{Input}:
        input points $X\subseteq \cX$,
        initial centers $C_0 \subseteq \cY$,
        number of iterations $t$,
        candidate selection procedure $\select$,
        tolerance $\tau \geq 0$.

      \item[] For $i = 1,2,\dotsc,t$:
        \begin{itemize}[itemsep=0ex]
          \item[$\bullet$]
            Choose candidate centers $Y_i := \select(X, C_{i-1})$.
          \item[$\bullet$]
            Set $C_i := C_{i-1} \cup \cbr[0]{c_i}$ for any $c_i \in Y_i$ that
            satisfies
            \[
              \phi_X\del[1]{C_{i-1} \cup \cbr[0]{c_i}}
              \ \leq \
              (1+\tau) \cdot \min_{c \in Y_i}
              \phi_X\del[1]{C_{i-1} \cup \cbr[0]{c}}
              \,.
            \]
        \end{itemize}
      \item[] \textbf{Output}: $C_t$.
    \end{itemize}
    \vspace{-4mm}
  \end{framed}
  \caption{Greedy algorithm for general $k$-medians problems.}
  \label{fig:alg}
\end{figure}

The bounds on \greedy will depend on one of two conditions being satisfied on $(C_{i-1}, Y_i)$ in each
round, at least with some probability.  These conditions are parameterized by an approximation factor $\gamma$.
In the sequel (e.g., generalized $k$-medians problems), the proofs will proceed by establishing one of these
conditions, and then directly invoke the guarantees on \greedy.

\begin{condition}
  \label{cond:core:1}
  For each $j \in [k]$, there exists $c \in Y_i$ such that
  $\phi_{\Aopt_j}(\cbr[0]{c}) \leq \gamma \cdot \phi_{\Aopt_j}(\cbr[0]{\copt_j})$.
\end{condition}

\begin{condition}
  \label{cond:core:2}
  There exists $c \in Y_i$ such that
  \begin{equation*}
    \max_{j\in[k]} \sbr{ \phi_{\Aopt_j}(C_{i-1}) - \min_{c\in Y_i}\phi_{\Aopt_j}(\{c\}) }_+
    \ \geq \
    \max_{j\in[k]}
    \sbr{
      \phi_{\Aopt_j}(C_{i-1}) -
      \gamma \cdot \phi_{\Aopt_j}(\cbr[0]{\copt_j})
    }_+
    \,.
  \end{equation*}
\end{condition}

Note that \Cref{cond:core:1} implies \Cref{cond:core:2}.

\begin{theorem}
  \label{cor:main}
  Let $\veps > 0$ and $\alpha \geq \gamma \geq 1$ be given,
  along with initial clustering $C_0$ satisfying $\phi_X(C_0) \leq \alpha \cdot \phi_X(\Copt)$,
  and lastly \greedy chooses $c_i\in Y_i$ with $\tau = 0$.
  If either
  \begin{itemize}
    \item[(1)]
      \Cref{cond:core:1} or \Cref{cond:core:2}
      hold for $(C_{i-1}, Y_i, \gamma)$ in each round,
      and $t \geq k\ln (( \alpha - \gamma ) / (\gamma \veps))$;
      or
    \item[(2)]
      \Cref{cond:core:1} or \Cref{cond:core:2}
      hold for $(C_{i-1}, Y_i, \gamma)$ conditionally independently with probability at least $\rho >0$ in each round,
      and $t \geq \max \{ k\ln (( \alpha - \gamma ) / (\gamma \veps)) / (2\rho), 2\ln(1/\delta)/\rho^2 \}$ for some $\delta > 0$;
  \end{itemize}
  then
     $\phi_X(C_t)
    \leq
    \gamma \cdot (1+\veps) \cdot \phi_X(\Copt)$
  holds unconditionally under the assumptions (1) above, and with probability at least $1-\delta$ under assumptions (2).
\end{theorem}

The proof is an immediate consequence of the following more general \namecref{thm:main}.

\begin{lemma}
  \label{thm:main}
  If \Cref{cond:core:1} or \Cref{cond:core:2} are satisfied with some $\gamma\geq1$
  for $(C_{i-1}, Y_i)$ in each round
  and $\tau < 1/(k-1)$, then the set of representatives $C_t$
  returned by \greedy satisfies
  \begin{equation*}
    \phi_X(C_t)
    \ \leq \
    \del{ 1 - \frac1k }^s \cdot (1+\tau)^s \cdot \phi_X(C_0)
    + \del{ 1 - \del{1 - \frac1k}^s \cdot (1+\tau)^s } \cdot \gamma \cdot
    \frac{1+\tau}{1-(k-1)\tau} \cdot \phi_X(\Copt)
  \end{equation*}
  with $s = t$.  If instead \Cref{cond:core:1} or \Cref{cond:core:2} holds with probability
  at least $\rho$ conditionally independently across rounds, then this bound on $\phi_X(C_t)$ holds with probability at least $1-\delta$
  with $s = \lfloor t\rho - \sqrt{t \ln(1/\delta)/2}\rfloor$.
\end{lemma}
\begin{proof}
  First consider any pair $(C_{i-1}, Y_i)$ satisfying \Cref{cond:core:1} or \Cref{cond:core:2},
  which simply means \Cref{cond:core:2} holds.
  Then
  \begin{align*}
    \phi_X(C_{i-1}) - \frac{\phi_X(C_i)}{1+\tau}
    & \ \geq \
    \phi_X(C_{i-1}) - \min_{c \in Y_i} \phi_X(C_{i-1} \cup \cbr[0]{c})
    && \text{(definition of $C_i$)}
    \\
    & \ \geq \
    \max_{c \in Y_i}
    \max_{j \in [k]}
    \phi_{\Aopt_j}(C_{i-1}) - \phi_{\Aopt_j}(C_{i-1} \cup \cbr[0]{c})
    \\
    & \ \geq \
    \max_{j \in [k]}
    \sbr{
      \phi_{\Aopt_j}(C_{i-1}) -
      \min_{c \in Y_i} \phi_{\Aopt_j}(\cbr[0]{c})
    }_+
    \\
    & \ \geq \
    \max_{j \in [k]}
    \sbr{
      \phi_{\Aopt_j}(C_{i-1}) -
      \gamma \cdot \phi_{\Aopt_j}(\cbr[0]{\copt_j})
    }_+
    && \text{(\Cref{cond:core:2})}
    \\
    & \ \geq \
    \frac1k \sum_{j=1}^k
    \sbr{
      \phi_{\Aopt_j}(C_{i-1}) -
      \gamma \cdot \phi_{\Aopt_j}(\cbr[0]{\copt_j})
    }_+
    \\
    & \ \geq \
    \frac1k
    \del{
      \phi_X(C_{i-1}) - \gamma \cdot \phi_X(\Copt)
    }
    \,.
  \end{align*}
  Rearranging the inequality gives the recurrence inequality
  \begin{equation}
    \phi_X(C_i)
    \ \leq \
    \del{ 1 - \frac1k } \cdot (1+\tau) \cdot \phi_X(C_{i-1})
    + \frac{\gamma}{k} \cdot (1+\tau) \cdot \phi_X(\Copt)
    \,.
    \label{eq:thm:main:1}
  \end{equation}
  Now let $(B_0, \ldots, B_{s-1})$ denote the subsequence of $(C_0, \ldots, C_{t-1})$
  where the corresponding pairs $(C_{i-1}, Y_i)$
  satisfy \Cref{cond:core:1} or \Cref{cond:core:2}.
  Since $\phi_X$ can not increase on rounds where neither condition holds,
  it still follows that
  \begin{equation*}
    \phi_X(B_i)
    \ \leq \
    \del{ 1 - \frac1k } \cdot (1+\tau) \cdot \phi_X(B_{i-1})
    + \frac{\gamma}{k} \cdot (1+\tau) \cdot \phi_X(\Copt)
    \,,
  \end{equation*}
  and therefore
  \begin{align*}
    \phi_X(B_t)
    & \ \leq \
    \del{ 1 - \frac1k }^s \cdot (1+\tau)^s \cdot \phi_X(B_0)
    + \sum_{i=0}^{s-1} \del{ 1 - \frac1k }^i \cdot (1+\tau)^i
    \cdot \frac{\gamma}{k} \cdot \phi_X(\Copt)
    \,.
  \end{align*}
  If the conditions hold for every round, then $s=t$ and the proof is done since $\tau < 1/(k-1)$.
  Otherwise, it remains to bound $s$;
  but since the conditions hold on a given round with probability at least $\rho$
  conditionally independently of previous rounds, it follows by Azuma's inequality that
    $\Pr\sbr[0]{ s \leq t\rho - \sqrt{t\ln(1/\delta)/2} }
    \leq \exp\del[0]{-2t (\ln(1/\delta)/(2t))}
    \leq \delta$
  as desired.
\end{proof}

\paragraph{Connection to supermodular and submodular optimization.}
\Cref{cor:main} recovers the analysis of the standard greedy method for set
cover~\citep{johnson1974approximation}:
if $X$ is a set of points and $\cY$ is a family of subsets of $X$, then the choices
$\Delta(x, S) := 1 + \bfo[x \not \in S]$ and $Y_i = \cY = \select()$ satisfy
\Cref{cond:core:1} with $\gamma = 1$.
A valid cover has cost $n$; \greedy (with $C_0 = \emptyset$, $\veps=1/n$)
finds a valid cover with cardinality $\leq k\ln(n)$ when one of cardinality $k$
exists.

More generally, the behavior of \greedy on the objective $\phi_X$ --- when $Y_i
= \cY$ for each $i$ --- is well understood, because $\phi_X$ is a monotone
(non-increasing) supermodular function.
Indeed, the results of \citet{nemhauser1978analysis} show that monotonicity and
supermodularity of $\phi_X$, together, imply the key recurrence
\cref{eq:thm:main:1} in the proof of \Cref{thm:main} in the case $\gamma = 1$.
Typically,
this \greedy algorithm is analyzed for the clustering problem
by regarding the function $f(S) := \phi_X(\cbr[0]{c_0}) - \phi_X(S
\cup \cbr[0]{c_0})$ as a submodular objective to be
maximized~\citep[e.g.,][]{mirzasoleiman2013distributed}; here $c_0 \in \cY$ is
some distinguished center fixed \emph{a priori}.
The results from this form of analysis are generally incomparable to those
obtained in the present work.
More details are given in \Cref{sec:supermod}.

\section{Generalized $k$-medians problems}
\label{sec:gkm}

The results of this section will specialize $\cX$, $\cY$, $\Copt$, and $\Delta$ in the following two ways.

The first setting, \emph{generalized $k$-medians}, is as follows.
There is a single ambient space $\cX = \cY$ for data points and centers, and
every data point is a possible center $X \subseteq \cY$.
Associated with this space is a distance function
$\dist: \cY \times \cY \to \bbR_+$
satisfying symmetry $\Delta(x,y) = \Delta(y,x)$ and the triangle inequality $\dist(x,z) \leq \dist(x,y) + \dist(y,z)$.
Lastly define $\Delta(x,y) := \dist(x,y)^p$ for some real number exponent $p \geq 1$.

Secondly, the distinguished sub-cases of \emph{Euclidean $k$-medians} and
\emph{$k$-means} are as follows.
The ambient space $\cX = \cY = \bbR^d$ is $d$-dimensional Euclidean space.
For Euclidean $k$-medians, $\Delta(x,y) = \|x-y\|_2$; for $k$-means,
$\Delta(x,y) = \|x-y\|_2^2$.
Moreover, for $k$-means, it is assumed (without loss of generality) that the
reference solution $\Copt$ satisfies $\copt_j = \mu(\Aopt_j)$ for each
$j\in[k]$.

Associated with each generalized $k$-medians instance is a real number $q\geq 1$;
for $k$-means instances, $q=1$, but otherwise $q=p$.
Additionally, define the normalized cost
\[
  \psi_{A}(C) \ := \ \del{ \phi_A(C) / |A| }^{1/q} \,.
\]
This normalization is convenient in the proofs, but is not used in the main theorems.
Lastly, all invocations of \greedy in this section set the tolerance parameter as $\tau = 0$.

All results in this section will assume an $\alpha$-approximate initialization $C_0$.
There exist easy methods attaining this approximation guarantee for $\alpha=\cO(1)$ with $|C_0|=\cO(k)$,
for instance \kmpp (cf.~\Cref{fact:km++}).

The basic approximation guarantee will make use of the following data-dependent quantity:
\[
  \constcorelb
  \ := \ \max\cbr{ \min_{x\in \Aopt_j}
  \frac{\psi_{\cbr[0]{x}}(\cbr[0]{\copt_j})}{\psi_{\Aopt_j}(\cbr[0]{\copt_j})} : j \in [k], |\Aopt_j| > 0 }
  \,.
\]
Note that $\constcorelb \leq 1$ in general, but \constcorelb can easily be smaller; for instance
$\constcorelb = 0$ with finite metrics.

\begin{theorem}
  \label{fact:gkm}
  Consider an instance of the generalized $k$-medians problem in $(\cX,D)$ with exponent $p$.
  Let $\veps > 0$ be given, along with $C_0$ with $\phi_X(C_0) \leq \alpha(1+\constcorelb)^q \phi_X(\Copt)$ for some $\alpha \geq 1$.
  Suppose \greedy is run for
  $t\geq k \ln ((\alpha-1) / \veps)$ rounds with $X = Y_i = \select()$ in each round.
  Then the resulting centers $C_t$
  satisfy
  \[
    \phi_X(C_t) \ \leq \ (1+\constcorelb)^q (1+\veps) \phi_X(\Copt)
  \]
  where $\constcorelb \in [0,1]$.
\end{theorem}

The key to the proof is the following property of generalized $k$-medians problems,
which implies \Cref{cond:core:1} holds in every round of \greedy.
This inequality generalizes the usual bias-variance equality for $k$-means problems;
a similar inequality appeared without \constcorelb for a slightly less general setting
in \citep{arthur2007kmeans++}.

\begin{lemma}[{See also \citep[Lemmas 3.1 and 5.1]{arthur2007kmeans++}}]
  \label{fact:Cp}
  Let a generalized $k$-medians problem be given.
  Then for any $j\in [k]$ and $y \in \cY$,
  \[
    \psi_{\Aopt_j}(\cbr{y}) \ \leq \ \psi_{\Aopt_j}(\cbr{{\copt_j}}) +
    \psi_{\cbr{y}}(\cbr{\copt_j}) \,,
  \]
  and moreover
  \[
    \min_{y \in {\Aopt_j}} \psi_{\Aopt_j}(\cbr{y})
    \ \leq \ (1+\constcorelb) \psi_{\Aopt_j}(\cbr{{\copt_j}})
    \ \leq \ 2 \psi_{\Aopt_j}(\cbr{{\copt_j}}) \,.
  \]
\end{lemma}
\begin{proof}
  The second bound is implied by the first,
  since the choice $z := \argmin_{y\in {\Aopt_j}} \psi_{\cbr{y}}({\copt_j})$
  satisfies $\psi_{\cbr{z}}({\copt_j}) = \constcorelb \psi_{\Aopt_j}(\cbr{\copt_j})$ where $\constcorelb \leq 1$,
  and so
  \begin{align*}
    \min_{y \in {\Aopt_j}} \psi_{\Aopt_j}(\cbr{y})
    & \ = \
    \psi_{\Aopt_j}(\cbr{z})
    \ \leq \
    \psi_{\Aopt_j}(\cbr{{\copt_j}}) + \psi_{\cbr{z}}(\cbr{\copt_j})
    \ \leq \ (1+\constcorelb) \psi_{\Aopt_j}(\cbr{\copt_j})
    \ \leq \ 2 \psi_{\Aopt_j}(\cbr{\copt_j}) \,.
  \end{align*}

  For the first bound,
  the special case of $k$-means follows from the standard bias-variance equality (cf.~\Cref{fact:bias-var}).
  Otherwise, $q = p \geq 1$,
  and the triangle inequality for $\dist$
  together with Minkowski's inequality (applied in $\bbR^{|{\Aopt_j}|}$ with counting measure) implies
  \begin{align*}
    \phi_{\Aopt_j}(\cbr{y})^{1/p}
    & \, \leq \,
    \del{ \sum_{x\in {\Aopt_j}}\del{\dist(x,{\copt_j}) + \dist({\copt_j},y)}^p}^{1/p}
    \, \leq \,
    \del{ \sum_{x\in {\Aopt_j}}\dist(x,{\copt_j})^p}^{1/p}
    + \del{ \sum_{x\in {\Aopt_j}}\dist({\copt_j},y)^p}^{1/p}
    \,.
  \end{align*}
  Dividing both sides by $|\Aopt_j|^{1/p}$ gives the bound.
\end{proof}

\begin{proof}[Proof of \Cref{fact:gkm}]
  By \Cref{fact:Cp},
  every round of \greedy satisfies \Cref{cond:core:1} with $\gamma \leq (1+\constcorelb)^q$.
  The result now follows by \Cref{cor:main}.
\end{proof}

\subsection{Reducing computational cost via random sampling}
\label{sec:gkm++}

One drawback of setting $Y_i = X$ as in \Cref{fact:gkm} is computational cost:
\greedy must compute $\phi_X(C_{i-1}\cup \cbr{c})$ for each $c\in X$.
One way to improve the running time, not pursued here, is to speed up $\phi_X$
via subsampling and other approximate distance computations
\citep{andoni2008near,feldman2011unified}.
Separately, and this approach comprises this subsection: the size of $Y_i$ can be made independent of $|X|$.
This is achieved via a random sampling scheme similar to \kmpp \citep{arthur2007kmeans++}, but repeatedly sampling
many new centers given the same fixed set of prior centers.

This random sampling scheme also has a data-dependant quantity.
Unfortunately \constcorelb is unsuitably small in general,
as it only guarantees the existence of \emph{one} good center in $X$:
instead, there needs to be a reasonable number of needles in the hay.
To this end, given $\varepsilon > 0$, define
\begin{align*}
  \core(\Aopt_j; \kappa)
  & \ := \ \cbr{ x\in \Aopt_j : \psi_{\cbr{x}}(\cbr{\copt_j}) \leq \kappa \psi_{\Aopt_j}(\cbr{\copt_j}) }
  \,,
  \\
  \constcore
  & \ := \ \inf\cbr{ \kappa\geq 0 : \forall j \in [k] \centerdot |\core(\Aopt_j; \kappa)| \geq \veps |\Aopt_j| / (1+\veps)}
  \,,
  \\
  \cAopt_j
  & \ := \ \core(\Aopt_j; \constcore) \,.
\end{align*}
The quantity \constcore will capture problem adaptivity in the main bound below.
By \Cref{fact:gkm++:misc}, $\constcorelb \leq \constcore \leq (1+\veps)^{1/q}$.

\begin{theorem}
  \label{fact:gkm++}
  Let $\veps > 0$ and $\delta > 0$ be given,
  along with $C_0$ with $\phi_X(C_0) \leq \alpha(1+\constcore)^q(1+\veps)^{q-1} \phi_X(\Copt)$ for some $\alpha \geq 1$.
  Suppose \greedy is run for
  $t\geq 4k \ln ((\alpha  - 1)/ \veps) + 8 \ln(1/\delta)$ rounds where $Y_i$ is chosen by the following scheme:
  \begin{quote}
    \selectpp: return $4k((1+\veps)/\veps)^{q+4})$ samples according to $\Pr[x] \propto \Delta(x, C_{i-1})$.
  \end{quote}
  Then with probability at least $1-\delta$, the resulting centers $C_t$
  satisfy
  \[
    \phi_X(C_t) \ \leq \ (1+\constcore)^q (1+\veps)^q \phi_X(\Copt)
  \]
  where $\constcore \in [\constcorelb, (1+\veps)^{1/q}]$.
\end{theorem}

The key to the proof is \Cref{fact:gkm++:1},
showing $(C_{i-1}, Y_i)$ satisfies \Cref{cond:core:2} with high probability.
In order to prove this, the following tools are adapted from a
high probability analysis of \kmpp due to \citet{aggarwal2009adaptive};
the full proof of \Cref{fact:gkm++:misc} can be found in \Cref{sec:kmpp}.

\begin{lemma}
  \label{fact:gkm++:misc}
  Consider any iteration $i$.
  \begin{enumerate}
    \item
      $\constcorelb \leq \constcore \leq (1+\veps)^{1/q}$
    \item
      If $y \in \cAopt_j$,
      then $\psi_{\Aopt_j}(C_{i-1}\cup \{y\}) \leq (1+\constcore)\psi_{\Aopt_j}(\Copt)$.
    \item
      If $\psi_{\Aopt_j}(C_{i-1}) > (1+\veps)(1 + \constcore) \psi_{\Aopt_j}(\Copt)$,
      then every $y \in Y_i$ satisfies $\Pr[ y \in \cAopt_j | y \in \Aopt_j ] \geq (\veps/(1+\veps))^{q+3}/4$.
  \end{enumerate}
\end{lemma}

\begin{lemma}
  \label{fact:gkm++:1}
  Suppose $\phi_X(C_{i-1}) > \gamma(1+\veps)\phi_{X}(\cbr{\Copt})$
  where $\gamma := (1+\veps)^{q-1}(1+\constcore)^q$.
  Then every element $c\in Y_i$ as chosen by \selectpp in \Cref{fact:gkm++}
  satisfies \Cref{cond:core:2} with constant $\gamma$ with probability at least $(\veps/(1+\veps))^{q+4}/(4k)$.
\end{lemma}
\begin{proof}
  Fix any cluster $\Aopt_m$ satisfying
  \[
    m \ := \ \argmax_{j\in [k]} \phi_{\Aopt_j}(C_{i-1}) - \gamma
    \phi_{\Aopt_j}(\Copt) \,.
  \]
  This $\Aopt_m$ must satisfy $\phi_{\Aopt_m}(C_{i-1}) > \gamma(1+\veps)\phi_{\Aopt_m}(\Copt)$
  and thus $\psi_{\Aopt_m}(C_{i-1}) > (1+\veps)(1+\constcore)\phi_{\Aopt_m}(\Copt)$,
  since otherwise
  \[
    \phi_X(C_{i-1}) - \gamma \phi_X(\Copt)
    \ = \ \sum_j \phi_{\Aopt_j}(C_{i-1}) - \gamma \phi_{\Aopt_j}(\Copt)
    \ \leq \ 0 \,,
  \]
  a contradiction.

  Observe that the probability of sampling a center $c$ from $\Aopt_m$ is
  \begin{align*}
    \Pr[ c \in \Aopt_m ]
    & \ = \ \frac {\phi_{\Aopt_m}(C_{i-1}) - \gamma \phi_{\Aopt_m}(\Copt) + \gamma \phi_{\Aopt_m}(\Copt)}{\phi_X(C_{i-1})}
    \\
    & \ = \ \frac {\max_j\left(\phi_{\Aopt_j}(C_{i-1}) - \gamma \phi_{\Aopt_j}(\Copt)\right) + \gamma\phi_{\Aopt_m}(\Copt)}{\phi_X(C_{i-1})}
    \\
    & \ \geq \ \frac {k^{-1} \sum_j\left(\phi_{\Aopt_j}(C_{i-1}) - \gamma\phi_{\Aopt_j}(\Copt)\right) + \gamma\phi_{\Aopt_m}(\Copt)}{\phi_X(C_{i-1})}
    \\
    & \ \geq \ \frac 1 k \left( 1 - \frac {\gamma\phi_X(\Copt)}{\phi_X(C_{i-1})}\right)
    \ > \ \frac 1 k \left( 1 - \frac {1}{1+\veps}\right)
    \ = \ \frac {\veps}{k(1+\veps)} \,.
  \end{align*}
  Since additionally the conclusion of \Cref{fact:gkm++:misc} holds due to the above calculation,
  \[
    \Pr[ c \in \cAopt_m ]
    \ = \ \Pr[ c \in \cAopt_m | c \in \Aopt_m ] \Pr[ c \in \Aopt_m ]
    \ \geq \
    \frac {1}{4k}\del{\frac {\veps}{1+\veps}}^{q+4} \ =: \ p_0 \,.
  \]
  Moreover, by \Cref{fact:gkm++:misc}, $c \in \Aopt_m$ implies
  \[
    \phi_{\Aopt_m}(C_{i-1}\cup\{c\})
    \ \leq \ (1+\constcore)^q \phi_{\Aopt_m}(\Copt)
    \ \leq \ \gamma \phi_{\Aopt_m}(\Copt)\,,
  \]
  and \Cref{cond:core:2} holds with probability at least $p_0$ since
  \begin{align*}
    \max_{j\in[k]} \left[ \phi_{\Aopt_j}(C_{i-1}) - \min_{c\in Y_i}\phi_{\Aopt_j}(\{c\}) \right]_+
    \ \geq \ & \left[ \phi_{\Aopt_m}(C_{i-1}) - \min_{c\in Y_i}\phi_{\Aopt_m}(\{c\}) \right]_+
    &\textup{(since $m\in [k]$)}
    \\
    \ \geq \ & \left[ \phi_{\Aopt_m}(C_{i-1}) - \gamma \phi_{\Aopt_m}(\{\copt_m\}) \right]_+
    &\textup{(by choice of $Y_i$)}
    \\
    \ = \ & \max_{j\in[k]} \left[ \phi_{\Aopt_k}(C_{i-1}) -
    \gamma\phi_{\Aopt_j}(\{\copt_j\}) \right]_+
    &\textup{(by choice of $m$)}
    \,.
    & \quad\ \ \qedhere
  \end{align*}
\end{proof}

The proof of \Cref{fact:gkm++} concludes by noting the success probability of \Cref{cond:core:2} is boosted with repeated sampling,
and thereafter invoking \Cref{cor:main}.
\begin{lemma}
  \label{fact:cond_boost}
  If a single sample from distribution $\cD$ satisfies \Cref{cond:core:2}
  with probability at least $\rho > 0$,
  then sampling $\lceil 1 / \rho \rceil$ points iid from $\cD$ satisfies \Cref{cond:core:2}
  with probability at least $1-1/e$.
\end{lemma}

\begin{proof}[Proof of \Cref{fact:gkm++}]
  If $\phi_X(C_i) \leq \gamma(1+\veps)\phi_{X}(\Copt)$ for any $i$, then it holds for all $j\geq i$.
  Thus suppose $\phi_X(C_i) \leq \gamma(1+\veps)\phi_{X}(\Copt)$ for all $i$;
  by \Cref{fact:gkm++:1} and \Cref{fact:cond_boost} and since \Cref{cond:core:1} implies \Cref{cond:core:2},
  then \Cref{cond:core:2} holds in every iteration each with probability at least $1/2$,
  and the result follows by \Cref{cor:main}.
\end{proof}

\subsection{Approximation ratios close to one}
\label{sec:not-a-ptas}

The previous settings only achieved approximation ratio $1+\veps$ when $X$ and
$\Copt$ allowed it (e.g., when $\constcorelb$ and $\constcore$ were small).
This subsection will cover three settings, each with $\cX = \cY = \bbR^d$, and three
corresponding choices for \select giving $1+\veps$ in general.
The first method is for $k$-means $\Delta(x,y) = \|x-y\|_2^2$,
the second for vanilla Euclidean $k$-medians $\Delta(x,y) = \|x-y\|_2$,
and the last for $\Delta(x,y) = \|x-y\|^p$ for any norm and $p\geq 1$.
Throughout this subsection, it is required that
$\cA$ denotes an optimal solution.

First, in the special case of $k$-means,
it suffices for \select to return a \emph{single non-adaptive set} of size $n^{\lceil 1/\veps\rceil}$
in each round.

\begin{theorem}
  \label{fact:kmeans-1+eps}
  Consider an instance of the $k$-means problem.
  Let $\veps \in (0,1)$ be given, along with $C_0$ with $\phi_X(C_0) \leq
  \alpha(1+\veps) \phi_X(\Copt)$ for some $\alpha \geq 1$.
  Suppose \greedy is run for $t\geq k \ln ((\alpha-1)/\veps)$ rounds using a
  \select routine that always returns the same subset $\bar Y$ defined by
  $\bar Y := \cbr[0]{ \sum_{i=1}^{\lceil1/\veps\rceil} x_i / \lceil1/\veps\rceil
  : x_1, x_2, \dotsc, x_{\lceil1/\veps\rceil} \in X }$.
  Then the resulting centers $C_t$ satisfy
  \[
    \phi_X(C_t) \ \leq \ (1+\veps)^2 \cdot \phi_X(\Copt) \,.
  \]
\end{theorem}
\begin{proof}
  Lemma~\ref{fact:inaba} implies that \Cref{cond:core:1} is satisfied with the set
  $\bar Y$ and $\gamma = 1+\veps$.
  The theorem therefore follows from \Cref{cor:main}.
\end{proof}
The construction of the set $\bar Y$ from \Cref{fact:kmeans-1+eps}
crucially relies on the bias-variance decomposition available for squared
Euclidean distance (cf.~\Cref{fact:bias-var}).

Next consider the Euclidean $k$-medians case $\Delta(x,y) := \|x-y\|_2$.
Since the mean (as used in \Cref{fact:kmeans-1+eps}) minimizes
$z\mapsto \phi_{\Aopt_j}(\cbr[0]{z})$
for $k$-means, it is natural to replace this mean selection
with a more generic optimization procedure.
Additionally, by using a \emph{stochastic online} procedure,
there is hope of using $\poly(1/\veps)$ data points as in \Cref{fact:kmeans-1+eps}.

There is one catch --- the standard convergence time for stochastic
gradient descent (henceforth sgd) depends polynomially on the diameter of the
space being searched.  In order to obtain multiplicative optimality
as in \Cref{cond:core:1},
the diameter of the space must be related to the cost of an optimal cluster
$\phi_{\Aopt_j}(\Copt)$.
Fortunately, this quantity can be guessed with $n^{-3}$ trials.
\begin{lemma}
  \label{fact:guess-ball}
  Define the procedure
  \begin{quote}
    \guessball: uniformly sample center $y\in X$ and sizes $b, m \in [n]$,
    let $B$ denote the $b$ points closest to $y$, and return the triple
    $(y, m, B)$.
  \end{quote}
  For any subset $A \subseteq X$ with mean $c:=\mu(A)$,
  with probability at least $n^{-3}$, simultaneously:
  $\phi_{A}(\cbr{c}) \leq \phi_B(\cbr{y}) \leq (1+2^q) \phi_{A}(\cbr{c})$,
  and $m = |A|$,
  and $\Delta(y, c) \leq \phi_{A}(\cbr{c})/|A|$.
\end{lemma}

The proof of \Cref{fact:guess-ball} will be given momentarily,
but with that concern out of the way, now note the sgd-based \select.

\begin{theorem}
  \label{fact:select-sgd}
  Consider the case of Euclidean $k$-medians, meaning $\Delta(x,y) = \|x-y\|_2$.
  Define a procedure \selectsgd which generates $2 n^{3 + \lceil 1/\veps^2\rceil}$ iid
  samples as follows:
  \begin{quote}
    Set $r := \phi_B(\cbr{y})/m$ where $(y,m,B)$ are from \guessball.
    Perform $s:=\lceil 1/\veps^2 \rceil$ iterations of sgd (cf. \Cref{fact:sgd})
    starting from $y$,
    on objective function $w\mapsto \bbE_x \|x-w\|_2$
    constrained to the Euclidean ball of radius $r$ around $y$,
    and using step size $\eta := 2r / \sqrt{s}$.  Return the unweighted average of the sgd iterates.
  \end{quote}
  If \greedy is run with initial clusters $C_0$ with $\phi_X(C_0) \leq \alpha(1+16\veps)\phi_X(\Copt)$ for some $\alpha \geq 1$,
  the above \selectsgd routine,
  and $t \geq 4k(\ln((\alpha - 1)/\veps) + 8\ln(1/\delta)$,
  then with probability at least $1-\delta$,
  the output centers $C_t$ satisfy $\phi_X(C_t) \leq (1+\veps)(1+16\veps)\phi_X(\Copt)$.
\end{theorem}

The full proof is in \Cref{app:not-a-ptas}, but can be sketched as follows.  For $k$-medians,
subgradients have norm 1, and the guarantees on \guessball give
$\Delta(y,\copt_j)\leq \phi_{\Aopt_j}(\cbr{\copt_j})/|\Aopt_j|$,
so $\lceil 1/\veps^2\rceil$ sgd iterations suffice (cf. \Cref{fact:sgd}), if somehow the random data
points were drawn directly from $\Aopt_j$.  But $|\Aopt_j| \geq 1$, so all data points are drawn from it with
probability at least $n^{-\lceil1/\veps^2\rceil}$.
Note that this proof grants the \emph{existence} of a good sequence of $\lceil
1/\veps^2\rceil$ examples together with a good triple $(y,m,B)$,
thus another approach, mirroring non-adaptive scheme in \Cref{fact:kmeans-1+eps},
is to enumerate these $n^{3+\lceil\veps^2\rceil}$ possibilities, but process them with sgd rather than
the uniform averaging in \Cref{fact:kmeans-1+eps}.

\begin{proof}[Proof of \Cref{fact:guess-ball}]
  Fix any cluster $A$.
  This proof establishes the existence of a point $y\in X$ and a set $B\subseteq X$ of closest points which
  satisfies all required properties together with $m:=|A|$.
  Since one such triple exists, then the probability of sampling one uniformly at random
  is at least $n^{-3}$.

  Choose $y \in X$ with $\Delta(y,c) = \min_{x\in A}\Delta(x,c) \leq \phi_{A}(\cbr{c})/|A|$.
  Let $B$ denote the $|A|$ points closest to $y$ in $X$ (ties broken arbitrarily).
  By \Cref{fact:Cp},
  \[
    \phi_B(\{y\})
  \ \leq \ \phi_{A}(\{y\})
  \ \leq \ 2^q\phi_{A}(\{c\})\,.
  \]
  If it also holds that $\phi_B(\{y\}) \geq \phi_A(\{c\})$, the proof is complete.

  Otherwise suppose $\phi_B(\{y\})< \phi_A(\{c\})$,
  which also implies $|A| \geq 1$, since otherwise $|X| = 0 = \phi_A(\cbr{c}) = \phi_A(\cbr{y})$.
  Consider the process of iteratively adding to $B$ those points in $X\setminus B$
  which are closest to $y$,
  stopping this process at the first time when $\phi_B(\{y\}) \geq \phi_A(\{c\})$;
  it is claimed that this final $B$ also satisfies $\phi_B(\{y\}) \leq (1+2^q)\phi_A(\{c\})$.

  To this end, note that the penultimate $B'$ did not satisfy $B'\supseteq A$,
  since that would mean
  \[
    \phi_{B'}(\cbr{y}) \ \geq \ \phi_A(\cbr{y}) \ \geq \ \phi_A(\cbr{c}) \ > \ \phi_{B'}(\cbr{y})\,.
  \]
  As such, the final added point $v$ can be no further from $y$ than the
  furthest element of $A\setminus B'$, meaning by \Cref{fact:Cp}
  \[
    \phi_{\cbr{v}}(\cbr{y})
    \ \leq \ \max_{u\in A\setminus B'} \phi_{\cbr{u}}(\cbr{y})
    \ \leq \ \phi_{A}(\cbr{y})
    \ = \ |A| \del{\psi_{A}(\cbr{y})}^q
    \ \leq \ |A| \del{ 2\psi_{A}(\cbr{c})}^q
    \ \leq \ 2^q \phi_{A}(\cbr{c}) \,,
  \]
  thus
  \[
    \phi_B(\cbr{y})
    \ = \ \phi_{B'}(\cbr{y}) + \phi_{\cbr{v}}(\cbr{y})
    \ < \ \phi_A(\cbr{c}) + 2^q \phi_A(\cbr{c}) \,.
    \qedhere
  \]
\end{proof}

To close, note how to get $1+\veps$ with $\Delta(x,y) = \|x-y\|^p$ for a general
norm and exponent $p$.
Unfortunately, the number of samples used here is exponential in the dimension.

\begin{theorem}
  \label{fact:gkm-ball}
  Let a generalized $k$-medians problem with $\Delta(x,y) = \|x-y\|^p$ for some norm $\|\cdot\|$ and exponent $p\geq 1$ be given,
  along with
  $\veps \in (0,1)$ and $\delta > 0$,
  and initial clusters $C_0$ with $\phi_X(C_0) \leq \alpha(1+\veps) \phi_X(\Copt)$ for some $\alpha \geq 1$.
  Suppose \greedy is run for
  $t\geq 4k \ln ((\alpha - 1) / \veps) + 8 \ln(1/\delta)$ rounds where $Y_i$ is chosen by \selectball,
  a procedure which returns $\cO(n^3 \veps^{-qd/p})$ iid samples, each generated as follows:
  \begin{quote}
    Obtain $(y, m, B)$ from \guessball, and output a uniform random sample
    from the $\dist$-ball of radius $2(\phi_B(\cbr{y})/m)^{1/p}$ centered at $y$.
  \end{quote}
  Then with probability at least $1-\delta$, the resulting centers $C_t$
  satisfy $\phi_X(C_t) \leq (1+\veps)^2 \phi_X(\Copt)$.
\end{theorem}

The full proof appears in \Cref{app:not-a-ptas}, but is easy to sketch.
By \Cref{fact:guess-ball}, not only is a point $y\in \Aopt_j$ with $\Delta(y,\copt_j)\leq\phi_{\Aopt_j}(\cbr{\copt_j})/|\Aopt_j|$
in hand, but additionally an accurate estimate on $\Delta(y,\copt_j)$.
The chosen sampling radius is large enough to include points around $\copt_j$ which are all $(1+\veps)$ accurate,
and the probability of sampling this smaller ball via the larger is just the ratio of their volumes.
The result follows by boosting the probability via \Cref{fact:cond_boost}
and applying \Cref{cor:main}.

\section{Experiments with $k$-means}
\label{sec:exp}

Experimental results appear in \Cref{fig:exp}.  The table in \Cref{fig:exp:table} summarizes the improvement over \kmpp
by \greedy with \kmpp-inspired \selectpp (cf. \Cref{sec:gkm++}):
for each of 5 UCI datasets with 1000-20000 points, \kmpp and \greedy/\selectpp were run 10
times for $t \in \{10, 50\}$, and then a ratio of median and minimum performance was recorded.  Of course, while the experiment
is favorable, it is somewhat unfair as \kmpp requires less computation.  On the other hand, \greedy is more amenable to various speedups,
for instance it is trivially parallelized.

The curves in \Cref{fig:exp:plot} plot the (median) cost on \texttt{aba} as a function of the number of centers.
These plots indicate an area where the analysis in the present work may be improved.  Namely, the results
of \Cref{sec:gkm} require good initialization for the best bounds
(e.g., a naive analysis with $C_0 = \emptyset$ introduces logarithmic dependencies on interpoint distance ratios).
According to \Cref{fig:exp:plot}, this is potentially an artifact of
the analysis: the dashed line uses \kmpp for the first 25 centers and \greedy for the remaining 25, and it does not outperform full \greedy.
Another possibility as that there are other data-dependent quantities which remove the need for good initialization.

\begin{figure}[t]
  \begin{subfigure}[b]{0.5\textwidth}
    \centering

    \begin{tabular}{|c||c|c||c|c|}
      \multicolumn{1}{c||}{}
      & \multicolumn{2}{c||}{\texttt{med(gr)/med(++)}}
      & \multicolumn{2}{|c}{\texttt{min(gr)/min(++)}}
      \\
      \hline
      \multicolumn{1}{c||}{}
      & $k = 10$
      & $k = 50$
      & $k = 10$
      & \multicolumn{1}{|c}{$k = 50$}
      \\
      \hline
      \texttt{aba} & 0.747 & 0.662 & 0.843 & 0.693 \\
      \texttt{car} & 0.794 & 0.915 & 0.837 & 0.924 \\
      \texttt{eeg} & 0.723 & 0.775 & 0.767 & 0.823 \\
      \texttt{let} & 0.746 & 0.787 & 0.855 & 0.804 \\
      \texttt{mag} & 0.729 & 0.788 & 0.824 & 0.811 \\
      \hline
      \hline
    \end{tabular}
    \caption{Cost ratio on five datasets.}
    \label{fig:exp:table}
  \end{subfigure}%
  \begin{subfigure}[b]{0.5\textwidth}
    \centering
    \includegraphics[width=0.95\textwidth]{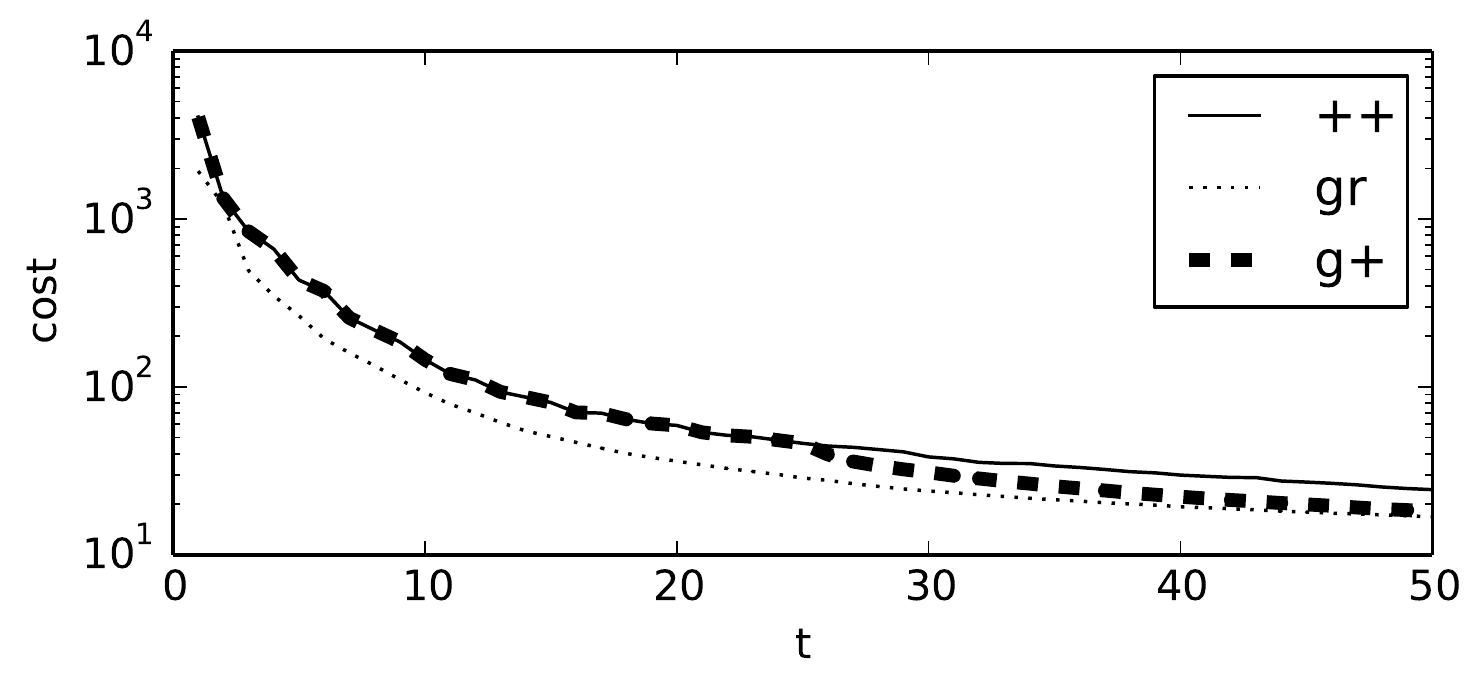}
    \caption{\kmpp initialization does not help.}
    \label{fig:exp:plot}
  \end{subfigure}%
  \caption{$k$-means comparisons between \kmpp (\texttt{++}), \greedy with subsampling (\texttt{gr}), and
  greedy initialized with \kmpp (\texttt{g+}).}
  \label{fig:exp}
\end{figure}

\subsection*{Acknowledgement}
The authors thank Chandra Chekuri for valuable comments and references.

\bibliographystyle{plainnat}
\bibliography{bib}

\appendix

\section{Technical tools}

\begin{lemma}[Bias-Variance]
  \label{fact:bias-var}
  For any finite subset $A\subseteq \bbR^d$
  and any $z\in\bbR^d$,
  \[
    \sum_{x\in A} \|x-z\|_2^2
    \ = \
    \sum_{x\in A} \del{ \|x-\mu(A)\|_2^2 + \|\mu(A) - z\|_2^2 }
    \,.
  \]
\end{lemma}
\begin{proof}
  \begin{align*}
    \sum_{x\in A} \|x-z\|_2^2
    & \ = \
    \sum_{x\in A} \|x-\mu(C) + \mu(C) - z\|_2^2
    \\
    & \ = \
    \sum_{x\in A} \del{ \|x-\mu(A)\|_2^2 + \|\mu(A) - z\|_2^2 }
    + 2 \ip{\mu(C) - z}{\sum_{x\in A} (x-\mu(A))} \,.
  \end{align*}
\end{proof}

\begin{lemma}[\citealp{inaba1994applications}]
  \label{fact:inaba}
  For any finite subset $A\subseteq \bbR^d$ and any $\eps > 0$, there exists
  $x_1, x_2, \dotsc, x_m \in A$ with $m = \lceil1/\eps\rceil$ such that $\mu_m
  := \sum_{i=1}^m x_i / m$ satisfies
  \begin{align*}
    \sum_{x \in A} \norm{x - \mu_m}_2^2
    & \ \leq \
    (1+\eps)
    \sum_{x \in A} \norm{x - \mu(A)}_2^2
    \,.
  \end{align*}
\end{lemma}
\begin{proof}
  This is a simple consequence of the first moment method and
  \Cref{fact:bias-var}.
\end{proof}

\begin{lemma}
  \label{fact:technical:1}
  If $x\in [0,1]$ and $p \geq 1$, then $(1+x)^{1/p} - 1 \geq x(2^{1/q} - 1)$.
\end{lemma}
\begin{proof}
  For convenience, define $f(x) := (1+x)^{1/p}-1$.  $f$ is concave, thus along $[0,1]$
  is lower bounded by its secant, which passes between $(0,f(0)) = (0,0)$ and $(1,f(1)) = (1,2^{1/p}-1)$.
\end{proof}

\begin{theorem}[Bernstein's inequality for martingales]
  \label{thm:bernstein}
  Let $(Y_i)_{i=1}^n$ be a bounded martingale difference sequence with respect
  to the filtration $\cF_0 \subset \cF_1 \subset \cF_2 \subset \dotsb$.
  Assume that for some $b,v>0$, $|Y_i| \leq b$ for all $i$ and $\sum_{i=1}^n
  \bbE\del{Y_i^2|\cF_{i-1}} \leq v$ almost surely.
  For all $\delta\in(0,1)$,
  \[
    \Pr\sbr{
      \sum_{i=1}^n Y_i
      > \sqrt{2v\ln(1/\delta)} + b\ln(1/\delta)/3
    }
    \ \leq \
    \delta
    \,.
  \]
\end{theorem}

\begin{theorem}[Convergence analysis for stochastic gradient descent]
  \label{fact:sgd}
  Consider the standard setup of stochastic gradient descent (sgd).
  \begin{itemize}
    \item
      Let convex function $f : \bbR^d \to \bbR$, reference point $\bar w$, and convex compact set $S$
      be given.
    \item
      Let a random subgradient oracle be given, which for any $w\in S$ returns (random) $\hat g$ with $\bbE(\hat g)\in \partial f(w)$.
    \item
      Let $L\geq 0$ be given which bounds the norms on full and stochastic gradients almost surely,
      meaning $\sup_{w\in S} \sup_{g\in \partial f(w)} \|g\|_2 \leq L$
      and $\|\hat g\|_2\leq L$ almost surely for any $\hat g$ returned by the oracle at any $w\in S$.
    \item
      Let $B\geq 0$ be given so that $\sup_{w\in S} \|w - \bar w\|_2 \leq B$.
    \item
      Let $w_1 \in S$ and total number of iterations $t$ be given, and suppose $w_{i+1} := \Pi_S(w_i - \eta \hat g_i)$ where
      $\eta := BL / \sqrt{t}$ and $\hat g_i$ is a stochastic gradient given by the oracle at $w_i$,
      and $\Pi_S$ is orthogonal projection onto $S$.
  \end{itemize}
  Then with probability at least $1-1/e$, $ f(\bar w_t) \leq f(\bar w) + 4BL/\sqrt{t}$.
\end{theorem}

\section{Analysis based on supermodularity}
\label{sec:supermod}

The behavior of the standard greedy algorithm (i.e., \greedy where $Y_i = \cY =
\select()$) is well-understood in the context of minimizing monotone supermodular
set functions, and indeed, the objective function $\phi_X$ fits this bill.
To see that $\phi_X$ is supermodular, consider any $C \subseteq C' \subseteq \cY$
and $c \in \cY \setminus C'$.
Then
\begin{align*}
  \phi_X(C) - \phi_X(C \cup \{c\})
  & \ = \
  \sum_{x \in X}
  \sbr{ \min_{c' \in C} \Delta(x,c') - \Delta(x,c) }_+
  \\
  & \ \geq \
  \sum_{x \in X}
  \sbr{ \min_{c' \in C'} \Delta(x,c') - \Delta(x,c) }_+
  \ = \
  \phi_X(C') - \phi_X(C' \cup \{c\})
  \,.
\end{align*}
Monotonicity holds since $C\subseteq C' \subseteq \cY$ implies $\min_{c\in C} \Delta(x,c) \geq \min_{c'\in C'} \Delta(x,c')$ for every $x\in X$.

\subsection{Standard analysis of greedy algorithm}

\citet{nemhauser1978analysis} show that supermodularity of an arbitrary set
function $f \colon 2^{\cY} \to \bbR$ is equivalent to the following property: for
any $S \subseteq T \subseteq \cY$,
\begin{equation}
  \sum_{y \in T \setminus S}
  \del{
    f(S) - f(S \cup \cbr[0]{y})
  }
  \ \geq \
  f(S) - f(T)
  \,.
  \label{eq:supermod:alt}
\end{equation}
If $f$ is also monotone, then for any $S, S^\star \subseteq \cY$ (where $S$
denotes a current solution and $S^\star$ denotes an arbitrary reference
solution),
\begin{align}
  \max_{y \in S^\star \setminus S} f(S) - f(S \cup \cbr[0]{y})
  & \ \geq \
  \frac1{\abs{S^\star \setminus S}}
  \sum_{y \in S^\star \setminus S}
  \del{
    f(S) - f(S \cup \cbr[0]{y})
  }
  \notag
  \\
  & \ \geq \
  \frac1{\abs{S^\star \setminus S}}
  \del{
    f(S) - f(S \cup S^\star)
  }
  \label{eq:supermod:lb}
  \\
  & \ \geq \
  \frac1{\abs{S^\star}}
  \del{
    f(S) - f(S^\star)
  }
  \,,
  \notag
\end{align}
where~\cref{eq:supermod:lb} follows from \cref{eq:supermod:alt} with $T = S \cup
S^\star$.
This shows that a greedy choice of $y \in \cY$ to minimize $f(S \cup \cbr[0]{y})$
yields a reduction in objective value at least as large as $(f(S) - f(S^\star))
/ \abs{S^\star}$.

Using $f = \phi_X$, $S = C_{i-1}$, $S^\star = \Copt$, and the fact that
$\phi_X(C_i) \leq (1+\tau) \cdot \min_{c \in \cY} \phi_X\del[1]{C_{i-1} \cup
\cbr[0]{c}}$ (as $Y_i = \cY = \select()$), the above inequality implies
\begin{equation}
  \phi_X(C_i)
  \ \leq \
  \del{ 1 - \frac1k } \cdot (1+\tau) \cdot \phi_X(C_{i-1})
  + \frac1k \cdot (1+\tau) \cdot \phi_X(\Copt)
  \,,
  \label{eq:supermod:recur}
\end{equation}
which exactly matches the key recurrence \cref{eq:thm:main:1} in the proof of
\Cref{thm:main} in the case $\gamma = 1$.

This clustering problem is more commonly viewed in the literature as a
submodular maximization problem~\citep[e.g.,][]{mirzasoleiman2013distributed}
with objective $f(S) := \phi_X(\cbr[0]{c_0}) - \phi_X(S \cup \cbr[0]{c_0})$,
where $c_0 \in \cY$ is some distinguished center fixed \emph{a priori}.
There, the same analysis of the greedy algorithm, after $t$ rounds starting with
$C_0 = \cbr[0]{c_0}$, yields a guarantee of the form
\begin{equation*}
  \phi_X(\cbr[0]{c_0}) - \phi_X(C_t \cup \cbr[0]{c_0})
  \ \geq \
  \del{ 1 - \frac1e }
  \max_{C \subseteq \cY\, : \; \abs{\cY} \leq t}
  \del{
    \phi_X(\cbr[0]{c_0}) - \phi_X(C \cup \cbr[0]{c_0})
  }
  \,,
\end{equation*}
which can be rewritten as
\begin{equation*}
  \phi_X(C_t \cup \cbr[0]{c_0})
  \ \leq \
  \frac1e
  \phi_X(\cbr[0]{c_0})
  + \del{ 1 - \frac1e }
  \min_{C \subseteq \cY\, :\; \abs{\cY} \leq t}
  \phi_X(C \cup \cbr[0]{c_0})
  \,.
\end{equation*}
This is generally incomparable to \Cref{cor:main}.

\subsection{Reducing computational cost via uniform random sampling}

For general monotone supermodular objectives $f$, random sampling can reduce
the computational cost of the standard greedy algorithm, although the
savings are modest compared to what can be achieved in the special case of $f =
\phi_X$ where additional structure is exploited.
This subsection describes such a ``folklore'' result~\citep[see,
e.g.,][Theorem 1.3]{buchbinder2015comparing}.

Consider a greedy choice of $y \in \bar Y$ to minimize $f(S \cup \cbr[0]{y})$,
where $\bar Y$, a multiset of size $m$, is formed by independently drawing
centers from $\cY$ uniformly at random.
Again, let $S^\star \subseteq \cY$ be an arbitrary reference solution.
Then
\begin{align*}
  \Pr\sbr{ \bar Y \cap (S^\star \setminus S) = \emptyset }
  & \ = \
  \del{ 1 - \frac{\abs{S^\star \setminus S}}{\abs{\cY}} }^m
  \,,
\end{align*}
and hence
\begin{align*}
  \bbE\sbr{
    f(S) - \min_{y \in \bar Y} f(S \cup \cbr[0]{y})
  }
  & \ \geq \
  \del{
    1 - \del{ 1 - \frac{\abs{S^\star \setminus S}}{\abs{\cY}} }^m
  }
  \frac1{\abs{S^\star \setminus S}}
  \sum_{y \in S^\star \setminus S}
  \del{
    f(S) - f(S \cup \cbr[0]{y})
  }
  \\
  & \ \geq \
  \del{
    1 - \del{ 1 - \frac{\abs{S^\star \setminus S}}{\abs{\cY}} }^m
  }
  \frac1{\abs{S^\star \setminus S}}
  \del{
    f(S) - f(S \cup S^\star)
  }
  \quad \text{(by~\cref{eq:supermod:lb})}
  \\
  & \ \geq \
  \del{
    1 - \exp\del{ - \frac{\abs{S^\star \setminus S}}{\abs{\cY}} \cdot m }
  }
  \frac1{\abs{S^\star \setminus S}}
  \del{
    f(S) - f(S^\star)
  }
  \\
  & \ \geq \
  \del{
    1 - \exp\del{ - \frac{\abs{S^\star}}{\abs{\cY}} \cdot m }
  }
  \frac1{\abs{S^\star}}
  \del{
    f(S) - f(S^\star)
  }
  \quad \text{(by convexity)}
  \,.
\end{align*}
If $m \geq (\abs{\cY}/\abs{S^\star})\ln(1/\delta)$, then the above inequality
implies
\begin{align*}
  \bbE\sbr{
    f(S) - \min_{y \in \bar Y} f(S \cup \cbr[0]{y})
  }
  & \ \geq \
  \frac{1-\delta}{\abs{S^\star}}
  \del{
    f(S) - f(S^\star)
  }
  \,.
\end{align*}
This leads to the same key recurrence as in \cref{eq:supermod:recur} (in
expectation) when $Y_i$ is formed by drawing $m \geq (\abs{\cY}/k)\ln(1/\delta)$
centers independently and uniformly at random from $\cY$.

A standard way to apply this technique to the generalized $k$-medians problem
described in \Cref{sec:gkm} (where typically one takes $\cY = \bbR^d$) is to
form each $Y_i$ by drawing $m$ centers independently and uniformly at random
from $X$.
The number of centers considered in each round of this \greedy implementation is
linear in $n = |X|$ (and this is also true of other methods studied
by~\citealt{buchbinder2015comparing}).
In contrast, the number of centers considered using the random sampling
technique from \Cref{sec:gkm++} is independent of $n$.

\section{\kmpp tools extracted from \citet{aggarwal2009adaptive}}
\label{sec:kmpp}

This appendix proves \Cref{fact:gkm++:misc}, the main tool in the results of \Cref{sec:gkm++}.
These tools are then used to prove adaptive (depending on \constcore) for vanilla \kmpp
in \Cref{fact:km++}.
The analysis is based on a high probability analysis of \kmpp due to
\citep{aggarwal2009adaptive},
merely simplified and adjusted to the setting here.

Throughout this section, the following additional notation will be convenient.

\begin{itemize}
  \item
    $C_{i-1}(z) := \argmin_{y \in C_{i-1}} \Delta(z, y)$.

  \item
    $\constgood := (1+\veps)(1+\constcore)$.

  \item
    $\conststop := (1+\veps)\constgood$.

  \item
    Split the optimal clusters into ``good'' and ``bad'' clusters,
    depending on how well they're ``covered'' by the centers in
    $C_{i-1}$.
    \begin{align*}
      \good_i & \ := \
      \cbr{
        j \in [k] : \psi_{\Aopt_j}(C_{i-1}) \leq \constgood \psi_{\Aopt_j}(\cbr{\copt_j})
      }, \\
      \bad_i & \ := \
      \cbr{
        j \in [k] : \psi_{\Aopt_j}(C_{i-1}) > \constgood \psi_{\Aopt_j}(\cbr{\copt_j})
      } \,.
    \end{align*}

\end{itemize}

\subsection{Proof of \Cref{fact:gkm++:misc}}

\begin{proof}[Proof of part 1 of \Cref{fact:gkm++:misc}]
  For the lower bound, given any $\kappa_0 < \constcorelb$,
  then there must exist $j\in [k]$ with $|A_k| > 0$ and $\min_{x\in\Aopt_j} \psi_{\cbr{x}}(\cbr{\copt_j}) > \kappa_0 \psi_{\Aopt_j}(\cbr{\copt_j})$,
  thus $|\core(\Aopt_j;b)| = 0$, and $\constcore > \kappa_0$ by definition of $\constcore$.  Since this holds for every $\kappa_0 < \constcorelb$,
  then $\constcore \geq \constcorelb$.

  For the upper bound,
  Set $\kappa_0 := (1+\veps)^{1/q}$,
  and consider any $j\in[k]$,
  setting $A_j := \core(\Aopt_j;\kappa_0)$ for convenience.
  The goal is to show $|A_j| \geq \veps |\Aopt_j| / (1+\veps)$,
  which implies the claim since $j$ is arbitrary and thus $\constcore \leq \kappa_0 = (1+\veps)^{1/q}$
  by definition of $\constcore$.

  The claim is trivial if $A_j = \Aopt_j$.
  If $A_j \subsetneq \Aopt_j$, then
  \[
    \phi_{\Aopt_j}(\cbr{\copt_j})
    \ \geq \ \sum_{x\in \Aopt_j\setminus \cAopt_j}\Delta(x, \copt_j)
    \ > \ \del{ |\Aopt_j| - |\cAopt_j| } \kappa_0^q \phi_{\Aopt_j}(\cbr{\copt_j}) / |\Aopt_j|
    \,.
  \]
  Rearranging, $|\cAopt_j| > |\Aopt_j| ( 1 - 1/\kappa_0^q) = \veps |\Aopt_j| / (1+\veps)$.
\end{proof}

\begin{proof}[Proof of part 2 of \Cref{fact:gkm++:misc}]
  By \Cref{fact:Cp}, every $y \in \cAopt_j$ satisfies
  \[
    \psi_{\Aopt_j}(C_{i-1} \cup \cbr{y})
    \ \leq \
    \psi_{\Aopt_j}(\cbr{y})
    \ \leq \
    \psi_{\Aopt_j}(\cbr{\copt_j})
    + \psi_{\cbr{y}}(\cbr{\copt_j})
    \ \leq \
    (1+ \constcore)(\psi_{\Aopt_j}(\cbr{\copt_j})) \,.
    \qedhere
  \]
\end{proof}

The proof of part 3 of \Cref{fact:gkm++:misc} will use the following lemma.

\begin{lemma}
  \label{fact:badcost}
  For any $j \in \bad_i$, and any $\hat{c} \in C_{i-1}$,
  \[
    \psi_{\cbr{\hat c}}(\cbr{\copt_j}) \ \geq \ (\constgood -
    1)\psi_{\Aopt_j}(\cbr{\copt_j}) \,.
  \]
\end{lemma}
\begin{proof}
  Take any $\hat{c} \in C_{i-1}$.
  Then, using the fact that $j \in \bad_i$ and \Cref{fact:Cp},
  \[
    \constgood (\psi_{\Aopt_j}(\cbr{\copt_j}))
    \ \leq \ \psi_{\Aopt_j}(\cbr{\hat c})
    \ \leq \ \psi_{\Aopt_j}(\cbr{\copt_j})
    + \psi_{\cbr{\hat c}}(\cbr{\copt_j}) \,.
  \]
  Rearranging gives the bound.
\end{proof}

Part 3 of \Cref{fact:gkm++:misc} is a consequence of the following more detailed bound.

\begin{lemma}
  \label{lemma:badcoremass}
  For any $j \in \bad_i$,
  \[
    \Pr\sbr{
      \hat{c}_i \in \cAopt_j  \,\Big\vert\, \hat{c}_i \in \Aopt_j
    }
    \ \geq \
    \frac{
      |\cAopt_j|
      \del{
        1 - \del{ \constcore  / (\constgood-1) }^{q/p}
      }^{p}
    }
    {
      |\Aopt_j|
      \del{ 1 + \frac {1}{\constgood - 1}}^q
    }
    \ \geq \
    \frac {1}{4} \del{\frac{\veps}{1+\veps}}^{q+3} \,.
  \]
\end{lemma}
\begin{proof}
  To start,
  \begin{align*}
    \Pr\sbr{
      \hat{c}_i \in \cAopt_j \,\Big\vert\, \hat{c}_i \in \Aopt_j
    }
    & \ = \ \frac{\phi_{\cAopt_j}(C_{i-1})}{\phi_{\Aopt_j}(C_{i-1})}
    \ = \
    \frac{|\cAopt_j|}{|\Aopt_j|}
    \del{
      \frac{
        \psi_{\cAopt_j}(C_{i-1})
      }{
        \psi_{\Aopt_j}(C_{i-1})
      }
    }^q
    \,.
  \end{align*}
  The proof proceeds by bounding the numerator and denominator separately.

  For the numerator, fix a particular $\tilde{x} \in \cAopt_j$,
  and observe
  \begin{align*}
    \lefteqn{
      \psi_{\cbr{\tilde x}}(C_{i-1}(\tilde x))^{q/p}
    } \\
    & \ \geq \
    \dist(\copt_j, C_{i-1}(\tilde x))
    - \dist(\copt_j, \tilde x)
    &\text{(triangle inequality of $\dist$)}
    \\
    & \ \geq \
    \dist(\copt_j, C_{i-1}(\copt_j))
    - \dist(\copt_j, \tilde x)
    \\
    & \ \geq \
    \dist(\copt_j, C_{i-1}(\copt_j))
    - \del{ \constcore \psi_{\Aopt_j}(\cbr{\copt_j}) }^{q/p}
    &
    \text{(since $\tilde{x} \in \cAopt_j$)}
    \\
    & \ \geq \
    \dist(\copt_j, C_{i-1}(\copt_j))
    - \del{ \constcore \psi_{\cbr{\copt_j}}(C_{i-1}(\copt_j)) / (\constgood - 1) }^{q/p}
    &\text{(\Cref{fact:badcost}, symmetry of $\dist$)}
    \\
    & \ = \
    \psi_{\cbr{\copt_j}}(C_{i-1}(\copt_j))^{q/p}
    \del{
      1 -
      \del{ \constcore  / (\constgood-1)}^{q/p}
    }
    \,,
  \end{align*}
  therefore
  \[
    \psi_{\cAopt_j}(C_{i-1})
    \ \geq \ \min_{\tilde x\in\cAopt_j} \psi_{\cbr{\tilde x}}(C_{i-1})
    \ \geq \
    \psi_{\cbr{\copt_j}}(C_{i-1}(\copt_j))
    \del{
      1 -
      \del{ \constcore  / (\constgood-1)}^{q/p}
    }^{p/q}
    \,.
  \]

  For the denominator
  \begin{align*}
    \psi_{\Aopt_j}(C_{i-1})
    & \ \leq \
    \psi_{\Aopt_j}(C_{i-1}(\cbr{\copt_j)})
    \\
    & \ \leq \
    \psi_{\Aopt_j}(\copt_j)
    + \psi_{\cbr{\copt_j}}(C_{i-1}(\copt_j))
    &\text{(\Cref{fact:Cp}, symmetry of $\dist$)}
    \\
    & \ \leq \
    \del{ 1 + \frac {1}{\constgood - 1}}
    \psi_{\cbr{\copt_j}}(C_{i-1}(\copt_j))
    &
    \text{(\Cref{fact:badcost})}
    \,.
  \end{align*}

  Combining the numerator and denominator bounds,
  \[
    \Pr\sbr{
      \hat{c}_i \in \cAopt_j  \,\Big\vert\, \hat{c}_i \in \Aopt_j
    }
    \ \geq \
    \frac{
      |\cAopt_j|
      \del{
        1 - \del{ \constcore  / (\constgood-1) }^{q/p}
      }^{p}
    }
    {
      |\Aopt_j|\del{ 1 + \frac {1}{\constgood - 1}}^q
    } \,.
  \]

  Lastly, to simplify the inequalities, first note $|\cAopt_j| / |\Aopt_j| \geq \veps / (1+\veps)$ by definition of \constcore.
  Next consider the case of $k$-means, implying $q=1 \neq 2 = p$.
  Then
  \[
    \Pr\sbr{
      \hat{c}_i \in \cAopt_j  \,\Big\vert\, \hat{c}_i \in \Aopt_j
    }
    \ \geq \
    \frac {\veps}{1+\veps}
    \del{
    \frac {\constgood - 1}{\constgood}
    }
    \del{
      1 - \del{ \constcore  / (\constgood-1) }^{1/2}
    }^{2}
    \,.
  \]
  To control these terms, note since $\constcore \geq 0$ that
  \[
    \frac {\constgood - 1}{\constgood}
    \ = \ 1 - \frac {1}{\constgood}
    \ = \ 1 - \frac {1}{(1+\veps)(1+\constcore)}
    \ \geq \ 1 - \frac {1}{1+\veps}
    \ = \ \frac {\veps}{1+\veps}
    \,.
  \]
  For the second term, since $\sqrt{\cdot}$ is concave, the tangent bound $\sqrt{a} \leq 1 + (a-1)/2$ holds,
  thus
  \[
    \del{
      1 - \del{ \constcore  / (\constgood-1) }^{1/2}
    }^{2}
    \ \geq \
    \frac 1 4
    \del{
      1- \del{ \constcore  / (\constgood-1) }
    }^{2}
    \ = \
    \del{ \frac {\veps}{2(1+\veps)} }^2
    \,.
  \]
  When the instance is not $k$-means, $q=p\geq 1$, and so
  \begin{align*}
    \Pr\sbr{
      \hat{c}_i \in \cAopt_j  \,\Big\vert\, \hat{c}_i \in \Aopt_j
    }
    & \ \geq \
    \frac {|\cAopt_j|}{\Aopt_j}
    \del{
    \frac {\constgood - 1}{\constgood}
    }^p
    \del{
      1 - \del{ \constcore  / (\constgood-1) }
    }^{p}
    \\
    & \ = \
    \frac {|\cAopt_j|}{\Aopt_j}
    \del{
      \frac{\constgood - 1 - \constcore}{\constgood}
    }^p
    \ = \
    \del{\frac{\veps}{1+\veps}}^{p+1}
    \,.
    \qedhere
  \end{align*}
\end{proof}

\subsection{\kmpp guarantee for generalized $k$-medians problems}

\begin{theorem}
  \label{fact:km++}
  Let $C_t$ denote the centers output by \greedy when run with $C_0 = \emptyset$
  and \select as in \Cref{fact:gkm++},
  except $|Y_i| = 1$ (e.g., only a single sample, as with \kmpp).
  With probability at least $1-\delta$,
  if $t \geq 8(k + \ln(1/\delta))((1+\veps)/\veps)^{q+4}$,
  then $\phi_X(C_t) \leq (1+\veps)^{2q}(1+\constcore)^q\phi_X(\Copt)$.
\end{theorem}

The proof uses the following \namecref{claim:badmass}.

\begin{lemma}
  \label{claim:badmass}
  In step $i$, at least one of the following is true.
  \begin{itemize}
    \item
      $\psi_X(S_{i-1}) \leq \conststop \psi_{X}(\cbr{\Copt})$,

    \item
      $\Pr\sbr{\hat{c}_i \in \bigcup_{j \in \bad_i} \Aopt_j}
      \geq 1 - \constgood/\conststop
      \geq \veps/(1+\veps)$,

  \end{itemize}
\end{lemma}
\begin{proof}
  Suppose $\phi_X(C_{i-1}) > \conststop \phi_{X}(\cbr{\Copt})$.
  Then
  \begin{align*}
    \Pr\sbr{\hat{c}_i \in \bigcup_{j \in \bad_i} \Aopt_j}
    & \ = \
    \frac{ \sum_{j \in \bad_i} \phi_{\Aopt_j}(C_{i-1}) }{
    \phi_X(S) }
    \\
    & \ = \ 1 -
    \frac{ \sum_{j \in \good_i} \phi_{\Aopt_j}(C_{i-1}) }{ \phi(C_{i-1}) }
    \\
    & \ \geq \ 1 -
    \frac{ \constgood^q \sum_{j \in \good_i} \phi_{\Aopt_j}(\cbr{\copt_j}) }
    { \conststop^q \phi_{X}(\Copt) }
    \geq 1 - \frac{\constgood^q}{\conststop^q}
    \ = \ \frac {\veps}{1+\veps}
    \,.
    \qedhere
  \end{align*}
\end{proof}

\begin{proof}[Proof of \Cref{fact:km++}]
  Consider the success events
  \[
    \cE_i \ := \
    \cbr{
      \phi_X(C_{i-1}) \leq \conststop \phi_X(\Copt)
      \ \vee \
      |\bad_i| = 0
      \ \vee \
      |\bad_{i+1}| < |\bad_i|
    }
  \]
  which states that at least one of the following statements is true
  upon choosing $c_i \in Y_i$ (note $\cbr{c_i} = Y_i$:
  \begin{enumerate}
    \item
      The approximation ratio $\phi_X(C_{i-1})/\phi_X(\Copt)$ before
      choosing $c_i$ is already at most $\conststop^q$.

    \item
      $\bad_i$ is empty, $\phi_X(C_{i-1}) \leq \constgood^q\phi_X(\Copt)$.

    \item
      The choice $c_i$ causes one bad set for $C_{i-1}$ to
      become good for $C_i$.

  \end{enumerate}
  After $k$ successes, $\bad_i$ is empty, thus it remains to control the number
  of stages before $k$ successes.
  By \Cref{fact:gkm++:misc} and since $\constgood \geq 1+\constcore$,
  if $c\in \cAopt_j$ for some $j\in \bad_i$,
  then $j\not\in \bad_{i+1}$.
  Therefore, by \Cref{claim:badmass} and \Cref{fact:gkm++:misc},
  \begin{align*}
    \Pr\sbr{\cE_i}
    & \ \geq \
    \Pr\sbr{
      c_i \in \bigcup_{j' \in \bad_i} \tilde{A}_{j'}
    }
    \\
    & \ = \
    \Pr\sbr{
      c_i \in \bigcup_{j' \in \bad_i} A_{j'}
    }
    \cdot
    \sum_{j \in \bad_i}
    \Pr\sbr{
      c_i \in A_j \,\Bigg\vert\,
      c_i \in \bigcup_{j' \in \bad_i} A_{j'}
    }
    \cdot
    \Pr\sbr{
      c_i \in \tilde{A}_j \,\Big\vert\, c_i \in A_j
    }
    \\
    & \ \geq \
    \frac 1 4 \del{ \frac {\veps}{1+\veps} }^{q+4}
    \ =: \ \rho
    \,.
  \end{align*}
  By Bernstein's inequality (cf.~\Cref{thm:bernstein}), letting $1_{\cE_i}$ denote the indicator random variable for $\cE_i$,
  $\Pr[ \sum_i 1_{\cE_i} \leq k ] \leq \delta$.
  As such, with probability at least $1-\delta$, $t$ iterations imply at least $k$ success
  amongst events $(\cE_i)_{i=1}^t$, meaning either $\phi_X(C_t)\leq \conststop^q\phi_X(\Copt)$ outright,
  or $\bad_t = \emptyset$ meaning $\phi_X(C_t) \leq \constgood^q\phi_X(\Copt)$.
\end{proof}

\section{Deferred proofs from \Cref{sec:not-a-ptas}}
\label{app:not-a-ptas}

First, the guarantee for \selectsgd.

\begin{proof}[Proof of \Cref{fact:select-sgd}]
  The result is a consequence of the following claim: any single center provided by \selectsgd
  satisfies \Cref{cond:core:1} with probability $n^{3+\lceil 1/\veps^2\rceil}/2$.
  Indeed, this suffices as with the proof of \Cref{fact:gkm++}: this probability can be boosted (e.g., via \Cref{fact:cond_boost}),
  and then \Cref{cor:main} completes the proof.

  To establish \Cref{cond:core:1}, fix any optimal cluster $\Aopt_j$, and any $\bar w$ output by \selectsgd.
  By \Cref{fact:guess-ball}, with probability at least $n^{-3}$, the estimate $\phi_B(\cbr{y})/m$
  satisfies $\|y - \copt_j\|_2 \leq \phi_{\Aopt_j}(\Copt)/|\Aopt_j| \leq \phi_B(\cbr{y})/m \leq 3\phi_{\Aopt_j}(\Copt)/|\Aopt_j|$.
  Moreover, with probability at least $(|\Aopt_j|/n)^{\lceil 1/\veps^2\rceil} \geq n^{-\lceil 1/\veps^2 \rceil}$,
  every data point randomly sampled during sgd was drawn from $\Aopt_j$.  Consequently,
  this sample is equivalent to one drawn directly from $\Aopt_j$ itself.
  Under this iid sampling, the function $f(c) := \bbE_x(\Delta(x, c)) = \phi_{\Aopt_j}(\{c\})/|\Aopt_j|$
  is
  convex with optimum $\copt_j$.

  To apply the bounds for sgd in \Cref{fact:sgd},
  the norms of iterates and subgradients (stochastic and full subgradients of $f$) must be controlled.
  Letting $S$ denote the ball of radius $r = \phi_B(\cbr{y})/m$ around $y$,
  namely the constraint set used by sgd within \selectsgd,
  every $w\in S$ satisfies
  \[
    \|w - \copt_j\|_2
    \ \leq \ \|w - y\|_2 + \|y - \copt_j\|_2
    \ \leq \ \phi_B(\cbr{y})/m + \phi_{\Aopt_j}(\Copt)/|\Aopt_j|
    \ \leq \ 4\phi_{\Aopt_j}(\Copt)/|\Aopt_j| \,.
  \]
  Moreover, for any $w\in S$ and any random $x\in\Aopt_j$ with $w\neq x$,
  the corresponding stochastic gradient has norm 1 since
  \[
    \left\| \frac {\partial}{\partial w} \|x - w\|_2 \right\|_2
    \ = \ \left\| \frac {\partial}{\partial w} \sqrt{\|x - w\|_2^2} \right\|_2
    \ = \ \frac {2\|x-w\|_2}{2\|x-w\|_2} \,.
  \]
  At the only point of non-differentiability, $x=w$, it suffices to control the magnitude of every directional derivative
  \citep[Theorem 23.2]{ROC},
  but by a similar calculation these are also 1.
  Combining these, both stochastic and full gradients of $f$ have norm 1, thus by \Cref{fact:sgd}, with probability at least $1-1/e$,
  the output $\bar w$ satisfies
  \[
    \phi_{\Aopt_j}(\{\bar w\})/|\Aopt_j| - \phi_{\Aopt_j}(\Copt)/|\Aopt_j|
    \ = \ f(\bar w) - f(\copt_j)
    \ \leq \ 16 \phi_{\Aopt_j}(\Copt) / (|\Aopt_j| \sqrt{\lceil 1/\veps^2 \rceil})
    \ \leq \ 16 \veps \phi_{\Aopt_j}(\Copt) / |\Aopt_j| \,.
  \]
  Since $\Aopt_j$ was arbitrary,
  \Cref{cond:core:1} is satisfied with $\gamma = 1 + 16\veps$ with probability at least $(1-1/e) n^{3+\lceil{1/\veps^2\rceil}}$.
\end{proof}

Lastly, the guarantees for \selectball.

\begin{lemma}
  \label{fact:gkm-ball:2}
  Given current centers $S$,
  suppose a single new center $\hat c$ is sampled according to the distribution in \selectball where $\veps \leq 1$.
  Then $S\cup \cbr{\hat c}$ satisfies \Cref{cond:core:1} with $\gamma=1+\veps$
  with probability $\Omega(n^{-3}\veps^{qd/p})$.
\end{lemma}
\begin{proof}
  Fix any reference cluster $\Aopt_j$.
  First consider the ball $B_r$ of radius (measured by $\dist$)
  $r:=(\gamma^{1/q} - 1)^{q/p} \psi_{\Aopt_j}(\cbr{{\copt_j}})^{q/p}$ around ${\copt_j}$;
  by \Cref{fact:Cp}, every $z\in B_r$ satisfies
  \[
    \psi_{\Aopt_j}(\cbr{z})
    \ \leq \ \psi_{\Aopt_j}(\cbr{{\copt_j}}) + \psi_{\cbr{z}}(\cbr{{\copt_j}})
    \ = \ \psi_{\Aopt_j}(\cbr{{\copt_j}}) + D(z,{\copt_j})^{p/q}
    \ \leq \ \gamma^{1/q} \psi_{\Aopt_j}(\cbr{{\copt_j}})
  \]
  and in particular $\phi_{\Aopt_j}(\cbr{z})\leq \gamma\phi_{\Aopt_j}(\cbr{{\copt_j}})$,
  meaning \Cref{cond:core:1} holds for this $j\in[k]$ with $\gamma$ for every $z\in B_r$;
  but $j\in [k]$ was arbitrary, so \Cref{cond:core:1} holds with probability exceeding the
  probability of the chosen point $\hat c$ falling within $B_r$.
  Furthermore, note by concavity of $(\cdot)^{1/q}$ the resulting tangent bound $\gamma^{1/q} -1 \leq (\gamma - 1)/q = \eps/q$,
  thus $r \leq \psi_{\Aopt_j}(\cbr{{\copt_j}})^{q/p}$.

  Now consider the triple $(y,m,B)$ returned by \guessball,
  and moreover the ball $B_R$ of radius $R := 2 \del{\phi_B(\cbr{y})/m}^{1/p}$ centered
  at $y$.
  By \Cref{fact:guess-ball} and the above upper bound $r \leq \psi_{\Aopt_j}(\cbr{{\copt_j}})^{q/p}$,
  every $z\in B_r$ satisfies
  \[
    D(z,y)
    \ \leq \
    D(z,\copt_j) + D(\copt_j, y)
    \ \leq \ r + \del{ \phi_{\Aopt_j}(\cbr{\copt_j})/|\Aopt_j| }^{1/p}
    \ \leq \ \psi_{\Aopt_j}(\cbr{\copt_j})^{q/p} + \psi_{\Aopt_j}(\cbr{\copt_j})^{q/p}
    \ \leq \ R \,,
  \]
  meaning $B_r \subseteq B_R$.
  As such, the probability of hitting a point in $B_r$ with a uniform sample from $B_R$
  is the volume ratio of the two balls, which by \Cref{fact:guess-ball} and the
  secant lower bound $\gamma^{1/q} - 1 = (1+\veps)^{1/q} - 1 \geq \veps(2^{1/q} - 1)$ for $\veps\in [0,1]$ (cf.~\Cref{fact:technical:1})
  satisfies
  \[
    \Pr\sbr{ z \in B_r | z \in B_R }
    \ = \ \del{ \frac {r}{R} }^d
    \ \geq \ \del{\frac{(\gamma^{1/q} - 1)^{q/p} \psi_{\Aopt_j}(\cbr{\hat c})^{q/p}}{2(1+2^q)^{1/p} \psi_{\Aopt_j}(\cbr{\hat c})^{q/p}}  }^d
    \ \geq \ \del{\frac {\veps^{q/p}}{2\cdot 3^{q/p}} }^d \,.
  \]
  The result now follows by multiplying this success probability with the $n^{-3}$ success probability for \guessball (cf. \Cref{fact:guess-ball}).
\end{proof}

\begin{proof}[Proof of \Cref{fact:gkm-ball}]
  As with the proof of \Cref{fact:gkm++},
  the proof follows by combining \Cref{fact:gkm-ball:2} with \Cref{fact:cond_boost} and \Cref{cor:main}.
\end{proof}

\end{document}